\documentclass[sigconf]{acmart}
\acmConference[Asia CCS 2025]{ACM ASIA Conference on Computer and Communications Security}{25-29 August, 2025}{Ha Noi, Vietnam}
\usepackage{tikz}
\usepackage{appendix}
\usepackage{amsmath}
\usepackage{amsthm}
\usepackage[]{graphicx}
\usepackage{subfig}
\graphicspath{{./pics/}}
\usepackage{booktabs}
\usepackage{diagbox}
\usepackage{cases}
\usepackage{bm}
\usepackage{url}
\usepackage{multirow}
\usepackage{threeparttable}
\usepackage{fancyvrb}
\usepackage{tipa}
\newtheorem{theorem}{Theorem}
\newtheorem{claim}[theorem]{Claim}
\usepackage{url}

\copyrightyear{2025}
\acmYear{2025}
\setcopyright{cc}
\setcctype{by-nc}
\acmConference[ASIA CCS '25]{ACM Asia Conference on Computer and Communications Security}{August 25--29, 2025}{Hanoi, Vietnam}
\acmBooktitle{ACM Asia Conference on Computer and Communications Security (ASIA CCS '25), August 25--29, 2025, Hanoi, Vietnam}\acmDOI{10.1145/3708821.3736214}
\acmISBN{979-8-4007-1410-8/2025/08}
\author{Brian Choi}
\affiliation{%
  \institution{Johns Hopkins University}
  \country{}
}
\email{bchoi11@jhu.edu}

\author{Shu Wang}
\affiliation{%
  \institution{Palo Alto Networks, Inc.}
  \country{}
}
\email{shuwang@paloaltonetworks.com}

\author{Isabelle Choi}
\affiliation{%
  \institution{University of California, Los Angeles}
  \country{}
}
\email{isabellechoi11@g.ucla.edu}

\author{Kun Sun}
\affiliation{%
  \institution{George Mason University}
  \country{}
}
\email{ksun3@gmu.edu}
\begin{document}
\newcommand{\TN}{ChainMarks}
\title{\TN{}: Securing DNN Watermark with Cryptographic Chain}
\begin{abstract}
With the widespread deployment of deep neural network (DNN) models, dynamic watermarking techniques are being used to protect the intellectual property of model owners. However, recent studies have shown that existing watermarking schemes are vulnerable to watermark removal and ambiguity attacks. Besides, the vague criteria for determining watermark presence further increase the likelihood of such attacks. In this paper, we propose a secure DNN watermarking scheme named ChainMarks, which generates secure and robust watermarks by introducing a cryptographic chain into the trigger inputs and utilizes a two-phase Monte Carlo method for determining watermark presence. First, ChainMarks generates trigger inputs as a watermark dataset by repeatedly applying a hash function over a secret key, where the target labels associated with trigger inputs are generated from the digital signature of model owner. Then, the watermarked model is produced by training a DNN over both the original and watermark datasets. To verify watermarks, we compare the predicted labels of trigger inputs with the target labels and determine ownership with a more accurate decision threshold that considers the classification probability of specific models. Experimental results show that ChainMarks exhibits higher levels of robustness and security compared to state-of-the-art watermarking schemes. With a better marginal utility, ChainMarks provides a higher probability guarantee of watermark presence in DNN models with the same level of watermark accuracy.
\end{abstract}
\maketitle
\section{Introduction}

Deep learning has shown its potential in multiple intelligent systems such as autonomous transportation, automated manufacture, and intelligent healthcare. However, the design and implementation of deep neural networks (DNNs) typically require significant resources for data collection, training, validation, and testing~\cite{miikkulainen2019evolving}. Because developing and possessing DNN models can provide a significant advantage, adversaries are highly motivated to steal the models for unauthorized use or resale. Therefore, it is crucial to protect the intellectual property (IP) of DNN models to avoid potential infringement; otherwise, it could impede the widespread deployment of these models.

Digital watermarking techniques are promising for safeguarding the intellectual property of DNN models by embedding covert information within the network for future verification. 
Similarly to traditional watermarking schemes designed for multimedia content, static watermarking techniques have been proposed to embed watermarks into the static parameters of DNN models (e.g., model weights) that are not changed during operation~\cite{uchida2017embedding, li2021spread, chen2019deepmarks, tartaglione2021delving}.
However, static DNN watermark solutions imply a white-box model that needs to access the model parameters during verification, which is not practical to protect the intellectual property of DNNs.

To better protect the intellectual property of DNN models, researchers develop dynamic DNN watermarking techniques by deliberately training a DNN model on both the original dataset and a watermark dataset~\cite{adi2018turning, zhang2018protecting}. By creating backdoors into the DNN model, these solutions output specific labels for a set of crafted inputs (i.e., watermark triggers). The over-parameterization of DNN models allows to plant the additional trigger inputs, i.e., watermarks, without affecting the overall classification accuracy of the DNN models. Trigger inputs can take the form of abstract images, adversarial examples, or inputs unrelated to the original task.

However, existing dynamic DNN watermarking schemes still face two challenges, namely, \emph{vulnerability to multiple watermark attacks} and \emph{vague criteria on watermark presence}. 
First, they are vulnerable to watermark removal attacks and watermark ambiguity attacks. A recent study~\cite{lukas2022sok} reveals the existing watermarking schemes are vulnerable to watermark removal attacks based on input preprocessing, model modification, or model extraction. Even worse, they are unable to resist watermark ambiguity attacks that allow attackers to forge additional watermarks to claim false DNN model ownership. One existing defense~\cite{fan2019rethinking} attempts to defeat watermark ambiguity attacks by adding a secret called digital passport as an extra input to the DNN watermarking scheme; however, attackers may leverage adversarial learning to find out an alternative qualified digital passport. Moreover, this method is vulnerable to the watermark removal attacks that are based on input preprocessing.
Second, existing watermarking schemes suffer from the unclear criteria for determining the presence of watermarks.
The accuracy of empirical estimation method for the watermark decision threshold is limited in practice, as the threshold is dependent on various practical factors such as model classification properties, types of watermark patterns, and their probability distributions~\cite{lukas2022sok}. Also, the existing estimation method is incapable of calculating the threshold in certain scenarios, e.g., when the $p$-value is extremely small.

In this paper, we propose \TN{}, a secure DNN watermarking scheme that can generate secure and robust watermarks by introducing \emph{one-way cryptographic chain relationships} into the watermark trigger inputs and utilizing a \emph{two-phase Monte Carlo estimation method} to determine the watermark presence.
\TN{} can efficiently defeat both watermark removal attacks and watermark ambiguity attacks by designing trigger inputs as a cryptographic chain.
First, the noise-like/pseudo-random triggers perform better than other forms. They are far away from the data distribution of regular tasks and are hence robust against removal attacks via fine-tuning, model pruning, and input preprocessing.
Second, the one-way chain property can stop the back-propagation algorithm of adversarial machine learning, which is used by watermark ambiguity attacks where attackers generate fake watermarks by satisfying multiple constraints (e.g., accuracy requirement). However, attackers cannot include hash functions in their optimization constraints and hence can only
perform guessing attacks.
Third, the accuracy of watermarks might not be 100\% due to the model post-processing; thus, unclear criteria for watermark presence may increase the likelihood of attacks.
To ensure detection accuracy, we develop a two-phase Monte Carlo method to enhance the estimation of watermark presence, by considering the probability distributions of both model classification and watermark generation.

\TN{} consists of two main modules, namely, watermark generation and watermark verification. %
To generate watermarks, we first construct a watermark dataset that consists of trigger inputs and their target labels. 
The trigger inputs are generated as a cryptographically chained sequence by repeatedly applying a one-way hash function over a secret key, which is a random seed selected by the model owner. Then, the digital signature of model owner is transformed into a number in base $C$ (that is, the number of classes), where each digit is assigned as the target label of a trigger input according to sequential order. By applying our cryptographic mechanism, the trigger inputs and their target labels are interrelated with the owner's seed key and digital signature, respectively.
Therefore, an adversary is unable to apply optimization-based adversarial attacks, since the acquired trigger inputs can only satisfy either the consistency of predicted labels with target labels or the presence of cryptographical chain in trigger inputs, but not both. 
To embed the watermarks, we train a watermarked DNN model over both the original and the watermark datasets. 

To verify the watermarks, the model owner provides the seed key to the verifier (e.g., a trusted third party) for regenerating the cryptographically chained trigger inputs, which are then fed into the DNN model to obtain the predicted labels. The verifier then calculates the Hamming distance between the predicted labels and the target labels extracted from the digital signature and determines the presence of watermarks based on a decision threshold, which is calculated via a two-phase Monte Carlo approach. 
In the first phase, the classification probability distribution of a DNN model is obtained by an empirical estimation method. In the second stage, our aim is to determine the number of additional random inputs required until the first hit occurs in the classes with the classification probability of zero. Deriving this threshold is equivalent to obtaining the probability distribution for the number of matching labels. The proposed two-phase Monte Carlo method enables us to obtain more accurate bounds on the distribution, since our method uses the real output distribution of DNN models instead of directly modeling the output classes with empirical probabilities. This advantage becomes more significant when the $p$-value is extremely small, as existing methods derive an output probability of zero.

To evaluate the robustness and security of \TN{}, we conduct extensive experiments on the watermarked models trained on CIFAR-10 / CIFAR-100 datasets.
Experiments show that our embedded watermarks do not affect the test accuracy of the original tasks. Compared with four state-of-the-art dynamic DNN watermarking schemes, \TN{} shows higher robustness against 16 watermark removal attacks. Moreover, due to the introduction of a cryptographic chain, \TN{} can resist the watermark ambiguity attack, which can easily bypass other schemes.
By investigating the effects of $p$-values, our proposed threshold estimation method is applicable to the scenarios with smaller $p$-values (i.e., higher level of security). In addition, the marginal utility of \TN{} is higher than that of other existing schemes, providing a higher probability guarantee of the watermark presence in the DNN models with the same level of watermark accuracy.

In summary, we make the following contributions:
\vspace{-0.03in}
\begin{itemize}
\item We propose a secure DNN watermarking scheme, \TN{}, which introduces a cryptographic chain into trigger inputs to counter both watermark removal and watermark ambiguity attacks.
\item We propose a new two-phase Monte Carlo method to estimate the decision threshold for watermark presence, providing a more accurate estimation and applying to the scenarios with higher security level.
\item We conduct extensive experiments to prove the robustness and security of \TN{} by comparing it with four state-of-the-art watermarking schemes against 17 watermark attacks over different models.
\end{itemize}

\vspace{-0.05in}
\section{Background}

\begin{figure}[t]
    \centering
    \includegraphics[width=2.8in]{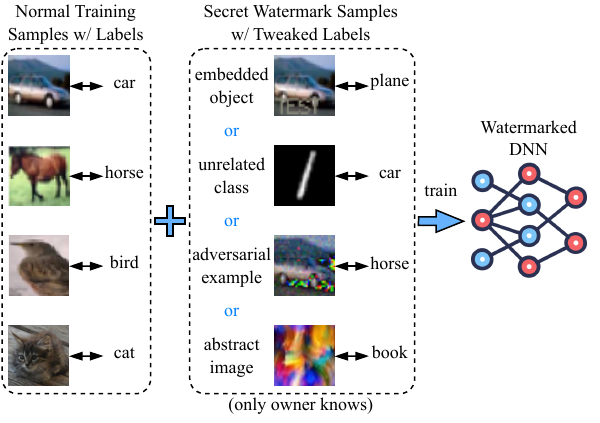}
    \vspace{-0.18in}
    \caption{An example of dynamic deep neural network watermarking scheme~\cite{adi2018turning, zhang2018protecting}.}
    \label{fig:scheme}
    \vspace{-0.18in}
\end{figure}

\subsection{Dynamic DNN Watermarking}

DNN watermarking schemes are essential for protecting intellectual property rights and ensuring the integrity of the model by embedding watermarks into the model behaviors in response to a crafted set of trigger inputs. During watermark verification, model behaviors can be observed to verify the presence of watermarks~\cite{adi2018turning, rouhani2019deepsigns, szyller2021dawn, le2020adversarial, zhang2018protecting, guo2018watermarking, zhang2020model}. 
In Figure~\ref{fig:scheme}, dynamic watermarking methods usually leverage DNN backdoors to generate trigger inputs, which are usually kept secret as they act as keys in the watermark embedding and verification processes.

The fundamental requirements of an effective watermarking technique include fidelity, generality, efficiency, robustness, and security~\cite{li2021survey}. Fidelity ensures that the embedded watermarks do not significantly impact the model accuracy. Generality states that a watermarking technique should be applicable to different DNN architectures and datasets.
The efficiency demands the overhead of watermark embedding and verification should be reasonable.
Robustness demands that embedded watermarks should be accurate even with model post-processing, e.g., fine-tuning~\cite{uchida2017embedding} and network pruning~\cite{zhu2017prune}, which can be performed by an entity with access to internal details of DNN models.
The security requirement ensures that a DNN watermarking scheme can withstand malicious attacks. An adversary may construct a surrogate model from the source model and then try to remove the embedded watermarks, replace the old watermarks with new ones, or fabricate fake watermarks to falsely claim ownership.
A secure dynamic DNN watermarking scheme should be able to defeat all forms of watermark attacks. 

\vspace{-0.05in}
\subsection{{One-Way Key Chain}}

A one-way key chain is constructed using a publicly known function $H$ that is easy to compute, but computationally hard to invert~\cite{lamport1981password}. Typically, the function $H$ is selected as a cryptographic hash function, e.g., MD5~\cite{rivest1992md5}, SHA-1~\cite{rijmen2005update}, or SHA-256~\cite{gilbert2003security}.
A one-way key chain of length $L+1$ is generated by iteratively applying $H$ to an initial key $K$ for $L$ times, resulting in $\{K, H(K), H^2(K), ..., H^L(K)\}$, where $H^i(K)$ denotes the hash value of $H^{(i-1)}(K)$, $2 \leq i \leq L$. We can compute $H^i(K)$ from $H^{(i-1)}(K)$; however, inferring $H^{(i-1)}(K)$ from $H^i(K)$ would be infeasible. The one-way key chain is widely used in authentication~\cite{rivest1992md5} and wireless sensor networks~\cite{tan2008secure}. 

\vspace{-0.04in}
\section{Threat Model}

We focus on protecting the intellectual property rights of DNN models via a secure dynamic watermarking technique. 
We assume that attackers have access to the model APIs to send queries and collect the outputs; thus, they may derive a surrogate model that approximates the watermarked model without knowledge of the secret watermark trigger inputs or target labels. 
A surrogate model with comparable accuracy to the source model effectively grants attackers access to a "white-box" version of the original model, including its parameters. 
Moreover, we assume that attackers have access either to unlabeled data from the same distribution~\cite{chen2019leveraging} or to labeled data from any distribution. 
We further assume there is a trusted third party, which serves as the verifier to ascertain if the claimed watermarks are present in the given model.
With access to the surrogate model, attackers may launch deep learning domain-specific attacks including watermark ambiguity attacks and watermark removal attacks.

\vspace{0.02in}
\noindent{\bf Watermark Ambiguity Attacks.} If an adversary successfully forges and embeds a second watermark into a watermarked model, there will be significant ambiguity with respect to model ownership, leading to false ownership claim. To compromise a watermarked model trained on the original dataset and trigger inputs, attackers can craft a new set of base trigger images with randomly assigned target labels. They then generate fake trigger images by adding trainable noise components, which have the same dimensions as the input images, to the base triggers. The attackers optimize a cross-entropy loss function between the target labels and the predicted labels of the fake trigger images. 
Given that attackers may construct surrogate models via transfer learning, existing DNN watermarking schemes remain vulnerable to watermark ambiguity attacks~\cite{guo2018watermarking, fan2019rethinking, lao2022deepauth, li2019piracy}. 
In particular, backdoor-based watermarking methods embed trigger inputs into the unused space of DNNs while sharing the same classifier for the original task, exhibit inherent limitations in resisting watermark ambiguity attack~\cite{fan2019rethinking}.

\vspace{0.02in}
\noindent{\bf Watermark Removal Attacks.}
A watermark removal attack involves using the source DNN model as input to generate a surrogate model as output. The objective is to eliminate embedded watermarks in the surrogate model while preserving a utility level comparable to that of the original model. In other words, the watermark retention rate in the surrogate model should be sufficiently low to prevent ownership claim with the source watermarks.
There are three types of removal attacks~\cite{lukas2022sok}. 
In input preprocessing attacks, an attacker with white-box access intentionally modifies data samples before passing them to the surrogate model during watermark verification.
In model modification attacks, an adversary with white-box knowledge may alter the internal parameters of the source model, typically through fine-tuning or pruning, to create a surrogate model.
Model fine-tuning is a transfer learning technique that adjusts an already trained model to perform another related task; however, the original model parameters are modified during fine-tuning, thus disrupting the embedded watermarks~\cite{chen2021refit}. 
Model pruning can set weights below a certain threshold to zero while maintaining the required model's accuracy; however, this process impact the watermark presence due to structural changes in the network. 
In model extraction attacks, an adversary trains a surrogate model by collecting input-output pairs from the source model, requiring only black-box access. 
The details of these three attack types are further discussed in Section~\ref{sec:watermark_removal}.

\vspace{-0.05in}
\section{\TN{} Design}

\subsection{System Overview}

Figure~\ref{fig:overview} shows the overall design of \TN{}. Given a DNN model $M_0$, the model owner first embeds watermarks into $M_0$ to obtain the watermarked model $M_w$.
Due to the white-box assumption, attackers can access both the architecture and weights of $M_w$. To remove existing watermarks or embed new watermarks, a pirate can perform watermark removal and watermark ambiguity attacks to deploy a surrogate model $M_s$. To protect intellectual property, the model owner can present watermark keys (i.e., the seed key to generate trigger inputs) to a trusted third party for verification of the claimed watermarks in $M_s$. Then the verifier calculates the watermark retention rate $R$ by evaluating if $m$ out of $L$ triggers are matched. If $R$ exceeds a retention threshold of $(1-\theta)$, the watermark is considered detected and ownership is successfully established; otherwise, the ownership claim fails. The detection threshold ($\theta$) depends on both the surrogate model ($M_s$) and the acceptable success probability for a guessing attack (i.e., $p$-value).

In \TN{}, we design watermarks based on a cryptographic chain, motivated by three hypotheses. 
First, one-way chain input can hinder the backpropagation algorithms used in adversarial machine learning, which underlie watermark ambiguity attacks. For example, watermarks ($x_1$, $x_2$) have labels ($y_1$, $y_2$). Attackers can apply optimization (adversarial ML) to find alternative inputs ($x_1'$, $x_2'$) that produce the same labels to falsely claim ownership. However, by introducing a chaining constraint, i.e., requiring $x_1'$, $x_2'$ satisfy $x_1' = hash(x_2')$, we dramatically increase the searching difficulty since adding a one-way function into optimization constraints is computationally infeasible.
Second, we observe that noise-like/pseudo-random triggers are more effective than other formats. This is because, in feature space, they are far from the distribution of natural data and hence are robust against fine-tuning/retraining. 
Third, the digital signature can be employed to verify that the triggers are truly generated and owned by the original model creator.

\begin{figure}[t]
    \centering
    \includegraphics[width=2.8in]{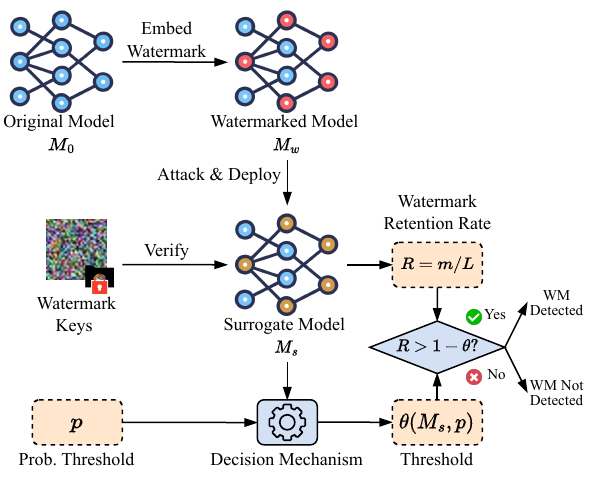}
    \vspace{-0.1in}
    \caption{The overview of our \TN{} schema.} 
    \label{fig:overview}
    \vspace{-0.15in}
\end{figure}

\vspace{-0.03in}
\subsection{Watermark Generation and Embedding}

All known dynamic watermarking schemes are vulnerable to watermark ambiguity attacks, due to the absence of a robust authentication mechanism for the trigger inputs.
To address this issue, we introduce additional cryptographic requirements for both the trigger inputs and the corresponding watermarked model.
Using our cryptographic chain, the feasibility of optimization-based ambiguity attacks is effectively eliminated.
In Figure~\ref{fig:embedding}, \TN{} establishes sequential cryptographic relationships for both the trigger inputs and the target labels. 

\vspace{0.03in}
\noindent {\bf Trigger Inputs.} 
The trigger inputs are generated as a cryptographically chained sequence by applying a cryptographic hash function over a seed key. In this sequence, each trigger inputs is ordered such that each input is derived from its predecessor by repeated application of the same hash function. Only the watermark owner, who possesses the secret key $K$, can generate the trigger inputs $\{B_{i}\}_{i=1}^{L}$, where $B_L = F(K)$ and $B_{i-1} = F(B_{i})$, for $i \in \{2,...,L\}$.

\vspace{0.03in}
\noindent {\bf Target Labels.} 
The target labels, which correspond to the ordered trigger inputs, are derived from the digital signature of the model owner. The digital signature, such as the hash value of the model owner's name, is converted into a number in base of the output dimension. For example, if the watermarked model has 8 classes, the digital signature is transformed into an octal number. Each digit of this number is then designated as the target label for a corresponding trigger input in sequential order. 
Thus, the generated sequence of target labels is denoted as $\{c_{i}\}$, where $i \in \{1,2,...,L\}$.

By applying our cryptographic mechanism, the trigger inputs and their corresponding target labels are cryptographically interrelated with the owner's seed key and digital signature, respectively.
Therefore, an adversary is unable to apply optimization-based adversarial attacks, e.g., watermark ambiguity attacks.

The seed key $K$, serving as a secret key, is randomly selected by the DNN model owner and provided to the hash function to construct a one-way chain of trigger inputs. 
Note that the trigger inputs generated in the chain are of the same size as the DNN inputs, whereas the seed key size is not constrained. 
In Figure~\ref{fig:embedding}, $F$ is an instance of the cryptographic one-way function; the sequence of watermark trigger inputs (or watermark keys) is generated by applying $F$, repeatedly.
Once a chain of size $L+1$ is constructed, the watermark keys can be disclosed and used in reverse order, meaning the last generated key will be the first to be used in the verification process. 
The number of disclosed keys is determined by the model owner based on specific applications.
Due to the one-way property of $F$, the security of the remaining undisclosed keys is still intact, allowing them to be used in subsequent rounds of watermark verification.
Thus, our method introduces cryptographic inter-constraints between the watermark trigger inputs, enabling multi-stage watermark verification using only a single seed key.

\begin{figure}[t]
    \centering
    \includegraphics[width=2.6in]{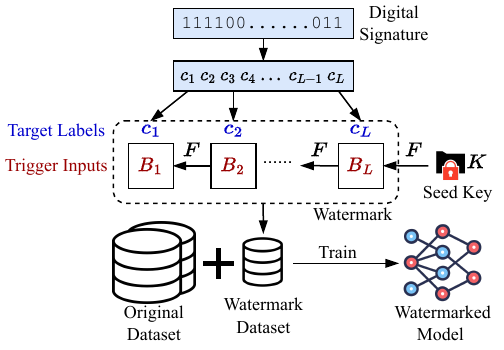}
    \vspace{-0.15in}
    \caption{Embedding a watermark into a DNN model. The model owner generates a chain by selecting a seed key $K$ and repeatedly applying a one-way function $F$ for $L$ iterations. The owner's digital signature is converted into target labels \{$c_i$\}, which are assigned to the triggers \{$B_i$\} in sequential order.}
    \label{fig:embedding}
    \vspace{-0.15in}
\end{figure}

In the watermark embedding process, the model owner’s digital signature is partitioned into target labels. 
Let the digital signature, denoted as $S$, be represented as a binary number with $|S|$ bits.
Since the number of classes supported by a DNN may not be a power of 2, the binary digital signature (in base 2) must first be converted into a number in base $C$, where $C$ represents the number of classes.  
For example, $C=10$ for CIFAR-10 and $C=100$ for CIFAR-100. 
Let $S_C$ denote the digital signature represented in base $C$.
The number of trigger input blocks is given by $L = \lceil \log_{C} S_{C} \rceil$.
Thus, the digital signature $S$ can be viewed as a sequence of $L$ digits in base $C$, e.g., $S = ({c_1} {c_2} {c_3} ... {c_L})_{C}$, where $0 \leq c_i < C$ for $1 \leq i \leq L$. These $c_i$ values are then used as the target labels for the watermark. Compared to directly dividing the binary signature $S$ into segments of length $\lceil \log_2 C \rceil$, our method allows for more trigger input blocks because
$\lceil \log_{C} S_{C} \rceil \geq \lceil \log_{2} S_{2} / \lceil \log_2 C \rceil \rceil$.

The watermark-embedded DNN is built by training a model from scratch on both the original dataset and a watermark dataset.
The watermark dataset consists of the trigger inputs and their corresponding target labels.
In Section~\ref{sec:efficiency}, we demonstrate that with \TN{}, the watermark embedding accuracy can achieve 100\%, without compromising the validation accuracy.
There are two possible implementation approaches to obtain the watermarked model. 
The first approach is to train the model on both the original data and the watermark data.
Another approach is to fine-tune a pre-trained DNN model with the watermark dataset.
We adopt the first method based on a reasonable assumption that the model owner can access both the original dataset and the watermark dataset.

\vspace{-0.03in}
\subsection{Watermark Verification}

Figure~\ref{fig:verification} illustrates the overview of watermark verification.
In the verification procedure, the model owner presents the seed key $K$ and uses it to re-generate the chain of $L$ trigger inputs, i.e., from $B_1$ to $B_L$.
These trigger input blocks are then fed into the DNN model to obtain the corresponding output labels, denoted as $c_1', c_2', ... , c_L'$.
These labels $\{c_i'\}$ can be concatenated to derive the digital signature.
\vspace{-0.08in}
\begin{equation}
    S' := c_1'\ ||\ c_2'\ ||\ ...\ ||\ c_{L-1}'\ ||\ c_{L}',
\end{equation}
\noindent where $c_i' \in \{0,1,...,C-1\}$ represents a number in base-$C$ numeral system. In addition, $c_i'$ is a digit of $S'$ in base $C$. 

Then, the Hamming distance between $S'$ and $S$, denoted as $d(S', S)$, is calculated as a measure of similarity between them. 
Hamming distance is a metric used to compare two data sequences of equal length, indicating the number of positions at which the corresponding symbols differ.
This metric is typically employed in error detection and error correction, particularly in coding theory and data communications over computer networks.

Another parameter in the proposed watermarking scheme is the decision threshold $\theta$ for the Hamming distance.
The verifier determines that the claimed watermark exists in the tested model only if $d(S', S) \leq \theta \cdot L$.
For example, if $\theta = 0.3$, the Hamming distance must be less than or equal to $0.3 \cdot L$ for $S'$ to be considered a match for $S$, indicating at least 70\% of the symbols must be identical.
Section \ref{sec:threshold} explains how to determine the threshold $\theta$ based on a given probability threshold, which corresponds to the success rate of simple watermark guessing attacks. This probability threshold can serve as a representative metric for the security of a model.

\begin{figure}[t]
    \centering
    \includegraphics[width=2.2in]{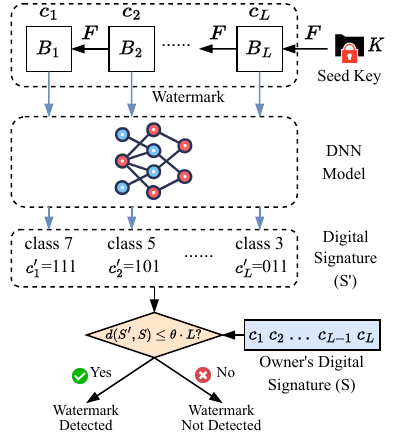}
    \vspace{-0.1in}
    \caption{The overview of watermark verification with a cryptographic side-classifier. $d(S',S)$ stands for the Hamming distance between two binary strings, $S'$ and $S$.}
    \label{fig:verification}
    \vspace{-0.15in}
\end{figure}

\vspace{-0.03in}
\subsection{Watermark Decision Threshold}\label{sec:threshold}

Watermark presence detection with Hamming distance is closely related to estimating the probability distribution of Hamming distances when a simple watermark guessing attack is performed.
In this scenario, an adversary can easily launch a simple guessing attack by providing a random seed key along with target labels derived from their digital signature, thereby attempting to claim the watermark’s presence.
By estimating the probability distribution of Hamming distances in such cases, we can derive the likelihood (or probability) of the watermark's presence given a particular seed key and target labels.
This distribution is influenced by the DNN model's actual classification probabilities for random inputs. To address this, we introduce a new dynamic Monte Carlo method that enables the derivation of tight bounds on the distribution.

\vspace{-0.03in}
\section{Two-Phase Monte Carlo Estimation}

To derive the probability distribution of the number of matching labels under watermark guessing attacks, we first need to obtain the DNN classification probability distribution for random inputs.
Then, the upper bounds of the matching probabilities can be derived by utilizing the classification probability distribution.
However, due to the complexity of obtaining an analytical solution directly from a DNN model, we propose a dynamic Monte Carlo method that employs a two-phase sampling approach to approximate the matching probability distribution.
In this section, we present a quantitative method for determining the detection threshold, addressing the question: given particular watermarks, how many retrieved watermarks are sufficient to assert ownership?

\vspace{-0.03in}
\subsection{Classification Distribution Estimation}

A significant challenge in estimating the classification probability distribution lies in the skewness across different classes. 
To address this, we conduct a simulation experiment and summarize the results in Table~\ref{tab:skewed_pr}.
Our findings reveals that the probability distribution exhibits a large standard deviation.
Specifically, after feeding 10 million random inputs into a Resnet-18 model trained on CIFAR-10, we observed that 5 out of 10 classes have a hit probability of 0.
Similarly, when 10 million random inputs are fed into the same model trained on CIFAR-100, 49 out of 100 classes are never hit.
To derive an approximate classification probability distribution for a DNN model, we define the following terms in this section:
\begin{itemize}
\item $N$: the total number of random inputs created in simulation.
\item $C$: the number of classes (or possible outputs) of the neural network. $C$ = 10 for CIFAR-10 and $C$=100 for CIFAR-100.
\item $\Gamma$: the set of all class indices, {$\{0, 1, 2, ..., C-1\}$}.
\item $n_i$: the number of inputs assigned to the class $i$, {$0 \leq i \leq C-1$}. 
\item $U=\{i_1, i_2, ..., i_k\}$: the set of $k$ class indices within {$[0,C-1]$} that are not hit by any of the $N$ inputs, with $|U|=k$.
\item $p_i$: the classification probability that an input falls into class $i$, calculated as $n_i / N$. 
\item $p_U$: the 0-hit probability, defined as $p_{i_1}+p_{i_2}+...+p_{i_k}$.
\end{itemize}

In the experiment shown in Table~\ref{tab:skewed_pr}, we obtain a set of class indices $U$, comprising 49 out of 100 indices for the DNN model trained over CIFAR-100.
However, the classification probabilities for the classes in $U$ cannot be directly estimated from the experiment data as no input samples are classified into these classes. 

To estimate the probability $p_U$, we devise a new approximation technique by modeling the experiment as a Bernoulli trial with a success probability of $p_U$ and a failure probability of $1-p_U$. 
To estimate the probability $p_U$ of the set $U$, we can calculate the expected number of random inputs required to observe the first successful classification into any class within $U$.
For example, in Table~\ref{tab:skewed_pr}, we define $U$ as the set of 49 class indices that receive no hits among the 10 million random inputs in the first stage.
That means that the number of random inputs $N$ is selected as 10 million to determine the initial sets $U$ and $\Gamma-U$. 
In the second stage, we generate additional random inputs until a class in $U$ is hit for the first time.
We perform 50 simulations and observe that on average 11,531,629 additional inputs are required to achieve this first hit.
Given that the number of trials before the first success in a Bernoulli process follows a geometric distribution, we can obtain $(1-p_U)/p_U=11,531,629$, which yields an estimated probability of $p_U = 1/(11,531,629+1) =8.672\times 10^{-8}$. 

After deriving $p_U$, the previous probabilities (calculated under the assumption of $p_U=0$), i.e., $p_i=n_i/N$ for all $i$ in $\Gamma-U$, need to be adjusted by applying a normalization technique.
Specifically, $p_i = p_i^{old} - p_U \cdot p_i^{old} = (1-p_U) \cdot p_i^{old}$, for all $i$ in $\Gamma-U$.

We also need to determine the approximate values of $p_i$ for all $i$ in $U$. Let $i_j$ denote the $j$-th index in $U$, for $1 \leq j \leq k$. Then, equation $\sum_{j=1}^{k} p_{i_j} = p_{U}$ must hold for the definition of $U$.

\vspace{-0.03in}
\subsection{Matching Probability Bound Estimation}
Our objective is to assess the probability that a randomly generated one-way input chain is erroneously accepted as a valid watermark under a Hamming distance threshold $\theta$.
A valid watermark chain should have at least $\lceil L \cdot (1-\theta) \rceil$ preserved blocks, where the labels produced by the DNN model match the digital signature of the legitimate owner of the model. 
The number of distinct combinations for selecting {$\lceil L \cdot (1-\theta) \rceil$} preserved blocks out of L blocks is
\begin{equation}
\left(
\begin{aligned}
\ L \ \ \\
\lceil L \cdot (1-\theta) \rceil
\end{aligned}
\right) = 
\left(
\begin{aligned}
\ L \ \ \\
\lfloor L \cdot \theta \rfloor
\end{aligned}
\right).
\end{equation}
\indent Given a DNN model and a sequence of $L$ target labels corresponding to the model owner’s digital signature, we aim to determine the success probability of a random guessing attack. 
Let $c_1, c_2, ..., c_L$ {($c_i \in [0,C-1]$ for $1 \leq i \leq L$)} denote the target classes (or labels) that match the owner's digital signature. 
A successful ownership claim is defined as achieving at least $m = \lceil L \cdot (1-\theta) \rceil$ matching labels between the attack inputs and the model owner's digital signature.
The process can be modeled as counting the successful matches in a sequence of $L$ independent yes/no experiments with the success probabilities of $p_{c_1}, p_{c_2}, ..., p_{c_L}$, where $c_i \in [0, C-1]$.
Poisson binomial distribution is well suited to calculate the above probabilities~\cite{hong2013computing}. 
Let $M$ be a random variable that indicates the number of matches. The probability of obtaining exactly $m$ matches out of $L$ chained random inputs can be expressed as
\begin{equation}
    Pr(M=m) = \sum_{A \in F_m} \prod_{i \in A} p_{c_i} \prod_{j \in A^{c}} (1-p_{c_j}),
    \label{eq:PrM}
\end{equation}
where $F_m$ is the set of all subsets of $m$ integers that can be selected from $\{1, 2, ..., L\}$. For example, if $L=3$ and $m=2$, then $F_2 = \{\{1,2\}, \{1,3\}, \{2,3\}\}$. 
For any subset $A \in F_m$, $A^c$ is the complement of $A$. For example, if $A = \{1,3\} \in F_{2}$, then $A^{c} = \{2\}$. 

Obviously, $F_m$ will contain $L!/((L-m)! \cdot m!)$ elements; hence, the Equation (\ref{eq:PrM})  is infeasible to compute in practice unless $L$ is small.
However, the following formula has been derived to approximate $Pr(M \geq m)$ using simple calculations.

\begin{table}[t]
    \centering
    \resizebox{0.95\linewidth}{!}{
    \begin{tabular}{c|c|c|c|c|c}
    \toprule
    \multirow{2}{*}{Dataset} & \multirow{2}{*}{\shortstack{Avg.\\Prob.}} & \multirow{2}{*}{\shortstack{Min\\Prob.}} & \multirow{2}{*}{\shortstack{Max\\Prob.}} & \multirow{2}{*}{\shortstack{Prob.\\Stdev}} & \multirow{2}{*}{\shortstack{\# of classes\\never hit}} \\
    {} & {} & {} & {} & {} & {} \\
    \midrule
    {CIFAR-10} & {0.1} & {0} & {0.9962} & {0.2987} & {5} \\
    {CIFAR-100} & {0.01} & {0} & {0.9433} & {0.0399} & {49} \\
    \bottomrule    
    \end{tabular}
    }
    \vspace{0.02in}
    \caption{Skewed probability distribution across different classes for DNN models trained on CIFAR-10/CIFAR-100.}
    \label{tab:skewed_pr}
    \vspace{-0.3in}
\end{table}

\begin{claim}
The probability of obtaining at least $m$ matches out of $L$ candidates is 
\begin{equation}
    {Pr}(M \geq m) \approx \Phi(\frac{L+0.5-\mu}{\sigma'}) - \Phi(\frac{m-0.5-\mu}{\sigma'}),
\end{equation}
\noindent where $\Phi$ is the cumulative distribution function (CDF) of the standard normal distribution, where the mean $\mu=L/C$ and the standard deviation $\sigma'$ follows the equation:
\begin{equation}
    \sigma' = \sqrt{\sum_{i=1}^{L} p_{c_i} - \sum_{i \in (\Gamma-U), 1 \leq i \leq L} p_{c_i}^2 - ({p_{U}^2}/{k})}.
\end{equation}
\end{claim}

\begin{proof}

An approximation technique~\cite{hong2013computing} can be employed to obtain the estimated probability. 
Approximation methods are still widely used due to their computational efficiency, especially when $L$ is large. 
We will utilize the normal approximation method, which is based on the central limit theorem (CLT).
If we define
\begin{equation}
    \mu = \sum_{i=1}^{L} p_{c_i},\ \sigma = \sqrt{\sum_{i=1}^{L} p_{c_i} (1-p_{c_i})}.
\end{equation}
Then, we can approximate the probability mass function using the normal approximation method with small errors for reasonably large values of $L (L \geq 10)$~\cite{choi2002approximating, volkova1996refinement}.
\begin{equation}
    Pr(M = m) \approx \varphi(\frac{m + 0.5 - \mu}{\sigma}),
\end{equation}
\noindent where $M$ denotes the number of matching labels (i.e., class indices) to the digital signature, and $\varphi$ stands for the \emph{probability distribution function (PDF)} of the standard normal distribution.

For example, with the threshold $\theta=0.3$, $L=100$ inputs, and $C=100$ classes, the attack success probability, i.e., the probability of a randomly generated chain of length $L$ yielding $m = \lceil L \cdot (1-\theta) \rceil = 70$ or more matches, would be
\begin{equation}
    \begin{aligned}
    P_{r}(M \geq 70) &= 1 - Pr(M < 70) \\
    &= 1 - \sum_{i=0}^{69} \varphi (\frac{i+0.5-\mu}{\sigma}) \\
    &= \sum_{i=70}^{L} \varphi(\frac{i+0.5-\mu}{\sigma}).
    \end{aligned}
    \label{eq:PMgeq30}
\end{equation}

Instead of using the \emph{PDF} of $\varphi(x)$, $P_r(M \geq m)$ may be obtained from $\Phi(x)$, which is the \emph{cumulative distribution function (CDF)} of the standard normal distribution.
\begin{equation}
    P_r(M \geq m) = \Phi(\frac{L+0.5-\mu}{\sigma}) - \Phi(\frac{m-0.5-\mu}{\sigma}).
\end{equation}

In our problem setting, we aim to approximate the probabilities $p_j$ for $j$ in $U$ to obtain the values of $\mu$ and $\sigma$, and then apply the above approximation formula.
Our goal is to find the upper bound on the attack success probability, whose example is shown in Equation (\ref{eq:PMgeq30}). 
The following formula holds for reasonable values of $\theta$:
\begin{equation}
    \mu \approx \frac{L}{C} \ll L \cdot \theta
\end{equation}

From the last formula in Equation (\ref{eq:PMgeq30}), we observe that an upper bound for the attack success probability $P_r'$ can be derived by finding a tight upper-bound $\sigma'$ for $\sigma$.
This is feasible because $\varphi(x)$ is a decreasing function for $x>0$.

The value of $\sigma$ can be rewritten as follows, noting that $|U|=k$.
\begin{equation}
    \begin{aligned}
        \sigma &= \sqrt{\sum_{i=1}^{L}(p_{c_i} - p_{c_i}^2)} = \sqrt{\sum_{i=1}^{L} p_{c_i} - \sum_{i=1}^{L} p_{c_i}^2} \\
        & = \sqrt{\sum_{i=1}^{L} p_{c_i} - \sum_{i \in (\Gamma - U), 1 \leq i \leq L} p_{c_i}^2 - \sum_{i \in U, 1 \leq i \leq L} p_{c_i}^2} \\ 
        & \leq \sqrt{\sum_{i=1}^{L} p_{c_i} - \sum_{i \in (\Gamma - U), 1 \leq i \leq L} p_{c_i}^2 - \sum_{i \in U, 1 \leq i \leq L} (\frac{p_U}{k})^2} \\
        & = \sqrt{\sum_{i=1}^{L} p_{c_i} - \sum_{i \in (\Gamma - U), 1 \leq i \leq L} p_{c_i}^2 - \frac{p_U^2}{k}} = \sigma'.
    \end{aligned}
    \label{eq:mu}
\end{equation}

Hence, the upper bound, $\sigma'$, of $\sigma$ is given in Equation (\ref{eq:mu}). 
The inequality in Equation (\ref{eq:mu}) holds due to the following optimization.
\begin{equation}
    \begin{aligned}
        &min \sum_{i \in U, 1 \leq i \leq L} p_{c_i}^2,\\
        &s.t.,\begin{aligned}
        &\sum_{i \in U, 1 \leq i \leq L} p_{c_i} = p_U,\\
        &0 \leq p_i \leq 1, 1 \leq i \leq L.
        \end{aligned}
    \end{aligned}
    \label{eq:min}
\end{equation}

In Equation (\ref{eq:min}), the item can achieve the minimum value \emph{only if} every $p_{c_i} = p_U/k$, where $i \in U$, $1 \leq i \leq L$.
\end{proof}
For instance, the margin of error (MOE) of the approximation formula in Claim 1 is calculated to be less than 1.2\% when compared to the precise values derived from Equation (\ref{eq:PrM}), for relatively small values (i.e., $L \leq 20$) using ResNet-18 models trained on CIFAR-10 and CIFAR-100 datasets.

\subsection{Threshold Decision}

The probability distribution in Table~\ref{tab:atk_suc_prb} shows the number of matches obtained by a basic watermark guessing attack against a ResNet-18 model trained on the {CIFAR-10} dataset. For the data in Table~\ref{tab:atk_suc_prb}, the probability distribution of regular data would not change much even after watermark embedding, due to the small proportion of watermark data (as Table~\ref{tab:acc} shows the accuracy only drops slightly).

\begin{table}[t]
    \centering
    \resizebox{\linewidth}{!}{
    \begin{tabular}{c|p{0.35in}<{\centering}|p{0.45in}<{\centering}|p{0.5in}<{\centering}|p{0.4in}<{\centering}|p{0.35in}<{\centering}}
    \toprule
    {$m$} & {0-6} & {7-8} & {9} & {10} & {11} \\
    \midrule
    {$Pr(M \geq m)$} & {1.0} & {0.9999} & {0.9984} & {0.8382} & {0.1618} \\
    \midrule
    \midrule
    {$m$} & {12} & {13} & {14} & {15-100} \\
    \midrule
    {$Pr(M \geq m)$} & {0.0015} & {4.01e-7} & {2.45e-12} & {0.0} \\
    \bottomrule
    \end{tabular}
    }
    \vspace{0.05in}
    \caption{The success probabilities over different match numbers for watermark guessing attacks against the ResNet-18 model trained on {CIFAR-10} ($C=10, L=100$).}
    \label{tab:atk_suc_prb}
    \vspace{-0.2in}
\end{table}

With the probability distribution, we set a threshold of the success probability for a simple watermark guessing attack.
Then, the success probability threshold can be mapped to the min match number, whereas we derive the decision threshold $\theta$.
For example, if the success probability threshold is set to $10^{-7}$ (i.e., the compromise probability should be less than $10^{-7}$), the number of matches $m$ should be at least 14 based on the probability distribution in Table~\ref{tab:atk_suc_prb}.
Then, we can derive the decision threshold of Hamming distance as $\theta = 1 - (m/L) = 0.86$.
Thus, when the chain length is set to $L=100$, the max tolerance error rate for a match achieves 86\%, i.e., a 14\% match in trigger inputs is sufficient for ownership claim. In practice, the match ratio is typically higher, e.g., a 90\% match can undoubtedly establish model ownership.
Therefore, the decision threshold for determining a watermark presence is dependent on both the watermarked model, $M_W$, and the success probability threshold, $p$. This threshold can be expressed as a function of $\theta(M_W, p)$.

Our proposed two-phase estimation method is more precise especially when the output probability distributions are skewed (i.e., small $p$-value), where the traditional one-phase estimation methods cannot work. It is because, for traditional estimation methods, the CDF function cannot accumulate quickly to reach the target probability sum $(1-p)$ due to the requirement of a large number of empirical estimations. Therefore, the state-of-the-art estimation method [32] only uses a moderate $p$-value with low confidence. Our method is applicable to use smaller $p$-values, which correspond to higher marginal utility and higher level of security (Section-\ref{sec:margin}).

\section{Experiments}

All watermarking schemes and removal attacks are implemented in PyTorch, on a server equipped with an NVIDIA GTX 1080 GPU.

\subsection{Datasets and DNN Models.}

Two image classification datasets, CIFAR-10 and CIFAR-100~\cite{Cifar}, are used to build watermarked DNN models. 
In our experiments, we use two model types, i.e., ResNet-18~\cite{he2016deep} and ResNet 28x10~\cite{zagoruyko2016wide}, which are trained on CIFAR-10 and CIFAR-100 datasets.
Model extraction attacks generally require extensive data access; thus, we assume adversaries have access to the full training dataset and know the architecture of the source model.

\vspace{-0.03in}
\subsection{Other DNN Watermarking Schemes}
We compare our scheme against four black-box watermarking methods~\cite{adi2018turning,zhang2018protecting}, summarized in Table~\ref{tab:bb_wm}.
In Adi et al.'s approach~\cite{adi2018turning}, watermark trigger inputs are abstract images paired with randomly assigned labels from the full class space. They explore two embedding strategies: training a model from scratch on a combined dataset (original + watermark data) and fine-tuning a pre-trained model. Their results show that training from scratch offers greater robustness to model modification attacks. Accordingly, we adopt this setting and reproduce their models using the combined dataset.

\begin{table}[t]
    \centering
    \renewcommand{\arraystretch}{1.05}
    \resizebox{0.975\linewidth}{!}{
    \begin{tabular}{c|c|c|c}
    \toprule
    {\bf Scheme} & {\bf Category} & {\bf Verification} & {\bf Capacity} \\
    \midrule
    {\TN{}} & {model dependent/independent} & {black-box} & {multi-bit} \\
    {Adi} & {model dependent/independent} & {black-box} & {multi-bit} \\
    {{Content}} & {model independent} & {black-box} & {zero-bit} \\
    {{Noise}} & {model independent} & {black-box} & {zero-bit} \\
    {{Unrelated}} & {model independent} & {black-box} & {zero-bit} \\
    \bottomrule
    \end{tabular}
    }
    \caption{Black-box watermarking schemes in evaluation.\label{tab:bb_wm}}
    \vspace{-0.3in}
\end{table}

Three other watermarking methods are proposed based on different types of trigger images: \emph{Content}, \emph{Noise}, and \emph{Unrelated} images~\cite{zhang2018protecting}. 
In the \emph{content}-based approach, trigger inputs are randomly chosen from a single class and modified with a fixed secret mask, such as a white square over a specific image region. 
The \emph{noise}-based method uses a mask generated from Gaussian noise, while the \emph{unrelated}-image approach selects trigger inputs from a domain unrelated to that of the original DNN. The procedures for generating target labels, embedding watermarks, and verifying ownership are similar to those of the \emph{Adi} scheme.

\vspace{-0.03in}
\subsection{Watermark Removal Attacks}
\label{sec:watermark_removal}
To compare the security and robustness of~\TN{} with four other schemes, we evaluate them against three categories of watermark removal attacks~\cite{lukas2022sok}: \emph{input preprocessing}, \emph{model modification}, and \emph{model extraction}, as summarized in Table~\ref{tab:WM_removal_atk}.

{\em 1) Watermark Removal via Input Preprocessing.}
These attacks remove watermarks without retraining the model by modifying the input images. For example, we can perform adaptive denoising~\cite{buades2005non} (or add Gaussian noise~\cite{zantedeschi2017efficient}) to the entire image.
The JPEG compression attack~\cite{dziugaite2016study} reduces image quality using JPEG encoding, which can eliminate watermarks.
In the input quantization attack~\cite{lin2019defensive}, pixel values are mapped to $2b$ evenly spaced intervals and replaced with the mean value of their interval. 
The input smoothing attack~\cite{xu2017feature} applies a mean, median, or Gaussian filter, resulting in a blurred image that may suppress watermark features.

{\em 2) Watermark Removal via Model Modification.}
These attacks alter the model itself to remove embedded watermarks.
Adversarial training~\cite{madry2017towards} improves model robustness by injecting adversarial examples, which are generated via Projected Gradient Descent~\cite{madry2017towards}, into the training set and fine-tuning the model on them using ground-truth labels. 
The fine-tuning attacks refer to a set of model stealing attacks that apply a transformation to the model by fine-tuning~\cite{uchida2017embedding}.
Fine-tuning attacks~\cite{uchida2017embedding} modify the model by fine-tuning to induce parameter changes. Four variants are considered: \emph{Fine-Tune All Layers (FTAL)} and \emph{Fine-Tune Last Layer (FTLL)} (with other layers frozen), both using ground-truth labels; and \emph{Retrain All Layers (RTAL)} and \emph{Retrain Last Layer (RTLL)}, which reinitialize either all layers or just the last layer and fine-tune using the model’s predicted labels.
Weight pruning~\cite{zhu2017prune} removes a random subset of weights from the model until a target sparsity level $\rho$ is reached.
Weight quantization~\cite{hubara2017quantized}, unlike input quantization, reduces the precision of model weights instead of input images.
Regularization attack~\cite{shafieinejad1906robustness} involves two phases: first, applying strong regularization to shift the model to a new set of parameters (potentially far from the original), which lowers test accuracy; second, recovering accuracy through fine-tuning.

\begin{table}[t]
    \centering
    \renewcommand{\arraystretch}{1.15}
    \resizebox{\linewidth}{!}{
    \begin{tabular}{c|c|c|c}
    \toprule
    {\bf Attack} & {\bf Category} & {\bf Param. Access} & {\bf Data Access} \\
    \midrule
    {Adaptive Denoising} & \multirow{4}{*}{\shortstack{Input\\Preprocessing}} & \multirow{9}{*}{White-box} & \multirow{3}{*}{None} \\
    {JPEG Compression} & {} & {} & {} \\
    {Input Quantization} & {} & {} & {} \\
    {Input Smoothing} & {} & {} & {} \\
    \cline{1-2}\cline{4-4}
    {Adversarial Training} & \multirow{6}{*}{\shortstack{Model\\Modification}} & {} & \multirow{5}{*}{Domain} \\ 
    {Fine-Tuning (RTLL, RTAL)} & {} & {} & {} \\
    {Weight Quantization} & {} & {} & {} \\
    {Weight Pruning} & {} & {} & {} \\
    {Regularization} & {} & {} & {} \\
    \cline{4-4}
    {Fine-Tuning (FTLL, FTAL)} & {} & {} & {Labeled Subset} \\
    \cline{1-4}
    {Transfer Learning} & \multirow{4}{*}{\shortstack{Model\\Extraction}} & \multirow{4}{*}{Black-box} & \multirow{4}{*}{Domain} \\
    {Retraining} & {} & {} & {} \\
    {Cross-Architecture Retraining} & {} & {} & {} \\
    {Adversarial Training (From Scratch)} & {} & {} & {} \\
    \bottomrule
    \end{tabular}
    }
    \caption{Watermark removal attacks in our evaluation.\label{tab:WM_removal_atk}}
    \vspace{-0.35in}
\end{table}

{\em 3) Watermark Removal via Model Extraction.}
Model extraction can remove watermarks by training a surrogate model to replicate the source model's behavior while discarding embedded watermarks~\cite{lv2024mea}.
In the retraining attack~\cite{tramer2016stealing}, a surrogate model is trained from scratch using input-label pairs obtained via API access.
The cross-architecture retraining variant uses a different model architecture to reduce watermark transfer. 
The transfer learning attack~\cite{torrey2010transfer} initializes the surrogate from a pre-trained model in another domain. 
In adversarial training (from scratch)~\cite{madry2017towards}, the surrogate is trained from scratch using adversarial examples, similar to standard adversarial training.

\vspace{-0.03in}
\section{Performance Analysis}

\subsection{Efficiency of Watermark Embedding}
\label{sec:efficiency}

We conduct 20 independent experiments to evaluate the test accuracy (i.e., accuracy on the original test dataset) of source models before and after watermark embedding, as well as the watermark accuracy (i.e., accuracy on watermark dataset) before and after watermark removal attacks.

Table~\ref{tab:acc} presents the average test/watermark accuracies under watermark ambiguity and 16 watermark removal attacks. 
After embedding the watermark using \TN{}, the watermark accuracy reaches 100\%, while the test accuracy experiences only a minor drop--0.8\% for models trained on CIFAR-10 (92.3\%$\rightarrow$91.5\%) and CIFAR-100 (69.1\%$\rightarrow$68.3\%).
This indicates that the impact of watermark embedding on model utility is negligible, typically under 1\%.
Table~\ref{tab:acc} further demonstrates that \TN{} achieves higher overall robustness than other methods while maintaining embedding efficiency comparable to \emph{Adi}. However, \TN{} is more secure and robust than \emph{Adi}. 
First, \emph{Adi} relies on multiple independent backdoor samples and is therefore susceptible to watermark ambiguity attacks when attackers generate adversarial alternatives for each backdoor trigger. In contrast, \TN{} can resist such ambiguity attacks via cryptographic chain (see Table~\ref{tab:robust}). Second, for the verification with small $p$-value, \emph{Adi} does not provide effective support, whereas \TN{} remains robust (see Section~\ref{sec:pvalues}).

After watermark removal/ambiguity attacks, the watermark accuracy decreases from 100\% to 67\% (34\%) on CIFAR-10 (CIFAR-100); however, the number of remaining valid watermarks is sufficient for ownership verification, based on the quantitative analysis in Section~\ref{sec:threshold}.
More analysis of watermark robustness under various attacks is provided in Section~\ref{sec:robust}.
In addition, a notable decline in test accuracy is observed when attackers perform watermark removal/ambiguity attacks. 
This indicates that adversaries cannot effectively reduce watermark accuracy in the surrogate model without significantly compromising its utility on the original task.
Besides, traditional watermarking methods improve embedding performance by injecting more watermark samples. 
However, \TN{} employs a hash function to ensure each watermark sample is cryptographically independent. Therefore,  adding more trigger samples (increasing chain length) does not lead to greater model memorization of the watermark.

\begin{table}[t]
    \centering
    \renewcommand{\arraystretch}{1.1}
    \resizebox{\linewidth}{!}{
    \begin{tabular}{c|c|c|c|c|c}
    \toprule
    \multirow{2}{*}{\bf Accuracy} & \multicolumn{5}{c}{\bf Accuracies (CIFAR-10/CIFAR-100)} \\
    \cmidrule{2-6}
    {} & {\bf \TN{}} & {\bf Adi} & {\bf Content} & {\bf Noise} & {\bf Unrelated} \\
    \midrule
    \multirow{2}{*}{\shortstack{Test Accuracy\\w/o WM embedding}} & \multirow{2}{*}{0.923/0.691} & \multirow{2}{*}{0.921/0.692} & \multirow{2}{*}{0.915/0.684} & \multirow{2}{*}{0.913/0.685} & \multirow{2}{*}{0.914/0.682} \\
    {} & {} & {} & {} & {} & {} \\
    \cmidrule{1-1}
    \multirow{2}{*}{\shortstack{Test Accuracy\\w/ WM embedding}} & \multirow{2}{*}{0.915/0.683} & \multirow{2}{*}{0.916/0.685} & \multirow{2}{*}{0.91/0.681} & \multirow{2}{*}{0.911/0.678} & \multirow{2}{*}{0.909/0.676} \\
    {} & {} & {} & {} & {} & {} \\
    \cmidrule{1-1}
    \multirow{2}{*}{\shortstack{Test Accuracy\\after Attack}} & \multirow{2}{*}{0.78/0.68} & \multirow{2}{*}{0.77/0.69} & \multirow{2}{*}{0.56/0.52} & \multirow{2}{*}{0.81/0.73} & \multirow{2}{*}{0.53/0.51} \\
    {} & {} & {} & {} & {} & {} \\
    \cmidrule{1-6}
    \multirow{2}{*}{\shortstack{WM Accuracy\\after Embedding}} & \multirow{2}{*}{1.0/1.0} & \multirow{2}{*}{1.0/1.0} & \multirow{2}{*}{1.0/1.0} & \multirow{2}{*}{1.0/1.0} & \multirow{2}{*}{1.0/1.0} \\
    {} & {} & {} & {} & {} & {} \\
    \cmidrule{1-1}
    \multirow{2}{*}{\shortstack{WM Accuracy\\after Attack}} & \multirow{2}{*}{0.67/0.34} & \multirow{2}{*}{0.69/0.37} & \multirow{2}{*}{0.58/0.33} & \multirow{2}{*}{0.73/0.41} & \multirow{2}{*}{0.64/0.35} \\
    {} & {} & {} & {} & {} & {} \\
    \bottomrule
    \end{tabular}
    }
    \caption{Test and watermark (WM) accuracy before/after watermark embedding and after watermark attacks.\label{tab:acc}}
    \vspace{-0.3in}
\end{table}

\vspace{-0.03in}
\subsection{Effects of \emph{p}-values}
\label{sec:pvalues}

To investigate the relationships between the required watermark accuracy ($1-\theta$) and the threshold probability $p$, we apply a two-phase Monte Carlo estimation method to analyze the \TN{} scheme.
For comparison, we adopt the empirical estimation approach proposed in~\cite{lukas2022sok} to evaluate four existing watermarking schemes.
For a range of threshold probabilities (i.e., $p$-values), we conduct 20 independent experiments to determine the average watermark accuracy necessary to successfully verify ownership.
The results are presented in Figure~\ref{fig:th_cifar} for both CIFAR-10 and CIFAR-100.




As shown in Figure~\ref{fig:th_cifar}, one notable observation is that, for small $p$-values, certain schemes encounter erroneous conditions where the cumulative distribution function (CDF) fails to reach the target probability mass.
This limitation arises from the insufficient number of models used in the empirical estimation of decision thresholds~\cite{lukas2022sok}, suggesting that the empirical method may not be reliable for the settings of small $p$-values.
Such issues are observed in the \emph{Noise}- and \emph{Content}-based schemes on CIFAR-10, and in all schemes except \TN{} on CIFAR-100.
Also, the \emph{Noise}-based scheme consistently exhibits high watermark accuracy requirements and may become impractical when the probability threshold $p$ is set below 0.01.

In Figure~\ref{fig:th_cifar}(a), for models trained on CIFAR-10, \TN{} exhibits the lowest required watermark accuracy among all evaluated schemes.
This is attributed to \TN{}’s high level of security, which enables a greater tolerance for errors, i.e., a higher allowable Hamming distance threshold $\theta$.
In Figure~\ref{fig:th_cifar}(b), for models trained on CIFAR-100, the required watermark accuracy of \TN{} is comparable to that of the \emph{Adi} and \emph{Unrelated} image-based schemes, while decision thresholds cannot be computed for the \emph{Noise}- and \emph{Content}-based schemes in most cases.
Besides, as the $p$-value increases, the required watermark accuracy decreases, since a higher $p$-value indicates a higher probability of compromise and hence requires fewer matched watermarks to verify ownership. 
Thus, \TN{} is able to meet higher security requirements.

\begin{figure}[t]
    \centering
    \subfloat[models on CIFAR-10.]{\includegraphics[width=0.5\linewidth]{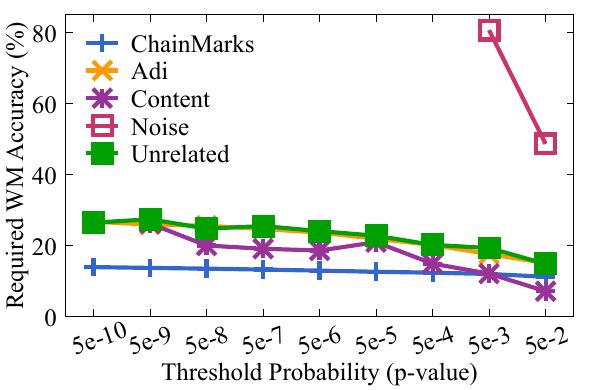}}
    \subfloat[models on CIFAR-100.]{\includegraphics[width=0.5\linewidth]{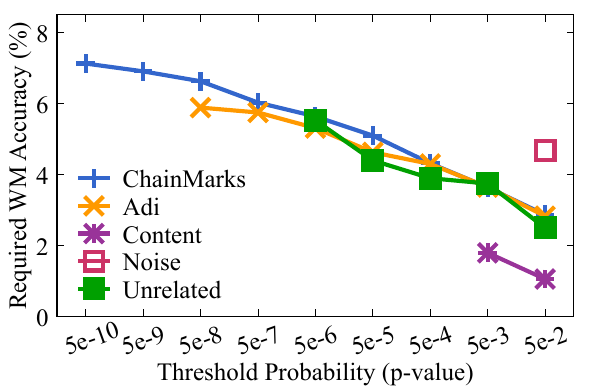}}
    \vspace{-0.15in}
    \caption{Required watermark accuracy ($1-\theta$) vs. threshold probability $p$, for different watermarking schemes.}
    \label{fig:th_cifar}
    \vspace{-0.2in}
\end{figure}



\vspace{-0.03in}
\subsection{Watermark Marginal Utility}
\label{sec:margin}

Comparing the required watermark accuracies across different schemes is inherently challenging due to the different decision thresholds. Therefore, relying only on the required watermark accuracy (or retention rate) is insufficient to evaluate the robustness of a watermarking scheme or the effectiveness of an attack.

An effective metric is needed to be derived to quantify the contribution of retained watermarks (i.e., watermark accuracy in a surrogate model) in terms of probabilistic guarantees on the success probability of random guessing attacks.
Therefore, we define this metric as \emph{watermark marginal utility}, which represents the average reduction factor in the attack success probability per unit increase in watermark accuracy (or decision threshold) within the surrogate model.
In other words, the metric indicates the degree of reduction that can be achieved in the threshold probability $p$ by increasing the watermark accuracy (or decreasing the decision threshold $\theta$) in a surrogate model. 

The watermark marginal utility is illustrated in Figure~\ref{fig:th_cifar} by dividing the ratio of $p$-values by the difference in watermark accuracy for two consecutive $p$-values on the x-axis.
If the $p$-value is reduced from $p_1$ to $p_2$ and the corresponding required watermark accuracy increases from $\tau_1$ to $\tau_2$, the watermark marginal utility can be calculated as ${({p_1}/{p_2})}/{(\tau_2 - \tau_1)}$. For example, if the $p$-value reduces from 0.05 to 0.005 and the corresponding watermark accuracy increases from 0.1123 to 0.1204, the estimated watermark marginal utility is computed as (0.05/0.005)/(0.1204 - 0.1123) = 1234.56.

Figure~\ref{fig:marginal_utility} presents the computed watermark marginal utilities for CIFAR10 and CIFAR100 models across different watermarking schemes.
The results show that \TN{} provides a higher watermark marginal utility compared to other schemes.
The marginal utility values for CIFAR-100 models are not provided for the \emph{Adi}, \emph{Content}, \emph{Noise}, and \emph{Unrelated}-based schemes due to the limitations in computing small $p$-values with empirical estimation method~\cite{lukas2022sok}.

\subsection{Overhead}

Training the model on CIFAR-10 for 200 epochs takes 2 hours, with a RAM usage of 2.5 GB and a GPU memory usage of 2.4 GB.
Training the model on CIFAR-100 for 200 epochs takes 7 hours, with a RAM usage of 3.7 GB and a GPU memory usage of 6.6 GB.
When watermark images are applied in the training set, the computational overhead remains negligible and does not significantly affect training time or memory usage.

\begin{figure}[t]
    \centering
    \includegraphics[width=2.2in]{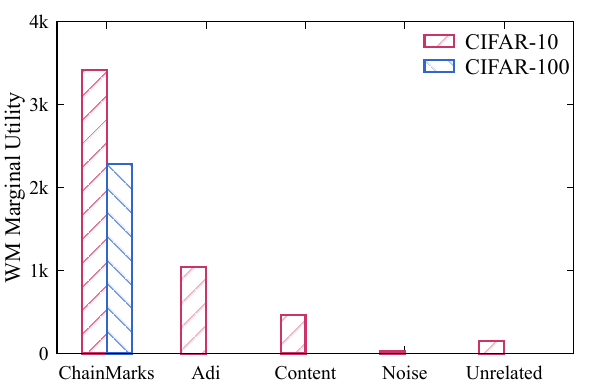}
    \vspace{-0.15in}
    \caption{Watermark marginal utility for various schemes.}
    \label{fig:marginal_utility}
    \vspace{-0.2in}
\end{figure}
\vspace{-0.03in}
\section{Security Analysis}

\subsection{Defeating Watermark Ambiguity Attacks}

The two cryptographic constraints introduced in \TN{} for trigger inputs and target labels render any optimization-based attacks, such as watermark ambiguity attacks~\cite{guo2018watermarking, fan2019rethinking, lao2022deepauth, li2019piracy}, infeasible.
In ambiguity attacks, attackers find adversarial watermarks by optimizing with the ``perturbed input - expected output'' pairs. 
When trigger inputs are independent (see Figure~\ref{fig:scheme}), such objectives are easily optimized since the added items are independent.
However, with cryptographic chaining, each trigger depends on a one-way hash function, which lacks gradients and thus obstructs backpropagation.
As illustrated in Figure~\ref{fig:optim_atk}, even when attackers inject trainable noise into fake triggers and optimize noise to match the output with digital signature, the optimized inputs break the required cryptographic chain, invalidating the watermark structure.

An adversary may also launch a guessing attack with trial and error. It can first choose a random seed key to create a one-way trigger input chain. Then, the adversary simply applies the original DNN model in a feedforward manner to check if the output labels match the claimed digital signature. Attackers can repeatedly attempt with different seed keys until a match is obtained. However, according to our analysis in Section~\ref{sec:threshold}, the success rate of random guessing is extremely low and can be determined by the selected threshold. Therefore, this approach will require an exponential number of trials, rendering it computationally infeasible.

\vspace{-0.03in}
\subsection{Countering Watermark Removal Attacks}
\label{sec:robust}

Because trigger inputs are derived from hash values, they can be regarded as random noise. 
Compared to the data distribution of primary task (e.g., image object recognition), these noise-like triggers are out-of-distribution with respect to both training and fine-tuning datasets. 
Thus, these triggers are robust against removal attacks, since fine-tuning typically alters model behavior within the task-specific feature space, leaving the trigger space unaffected.

Evaluating the success of an attack requires careful consideration of both test accuracy loss and watermark accuracy degradation in surrogate models.
To assess watermark robustness, we deploy five watermarking schemes, including~\TN{}, to construct watermarked models trained on the CIFAR-10/CIFAR-100 datasets.
As listed in Table~\ref{tab:robust}, we then apply 17 distinct attacks to the watermarked (source) models.
Let $M_s$ denote the surrogate model obtained by an attacker, and let $p$ represent the threshold probability used to derive the decision threshold $\theta$.
For an attack to be considered successful, it must meet both of the following criteria.

\begin{figure}[t]
    \centering
    \includegraphics[width=2.2in]{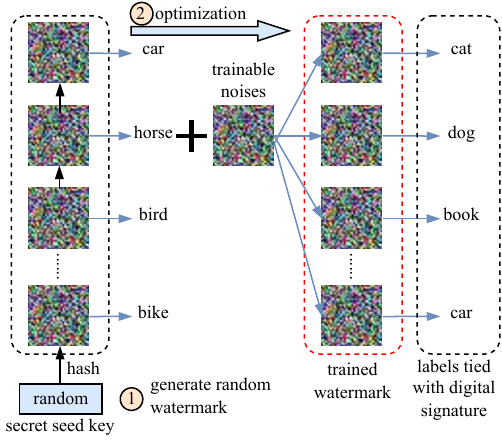}
    \vspace{-0.15in}
    \caption{Attackers can use optimization to generate adversarial trigger inputs with associated invalid signature labels, but the resulting inputs are no longer cryptographic chained.}
    \label{fig:optim_atk}
    \vspace{-0.18in}
\end{figure}

\noindent{\bf Test Accuracy Criteria.} 
To satisfy the accuracy drop threshold of 0.1, the test accuracy of the surrogate model should be at least 90\% of that of the watermarked model.

\noindent{\bf Watermark Accuracy Criteria.}
The watermark accuracy in $M_s$ should be less than the original decision threshold $\theta(M_W, p)$, where the threshold probability $p=0.01$.

We obtain the decision threshold $\theta(M_W, p)$ for each watermarking scheme.
For \TN{}, the threshold is derived using a two-phase Monte Carlo estimation method, as described in Section~\ref{sec:threshold}.
For the other four watermarking schemes, no precise method exists for computing $\theta$; therefore, we adopt the empirical estimation technique proposed in~\cite{lukas2022sok}.
This approach estimates the watermark accuracy of an unmarked model with two random variables. 
Specifically, we estimate the cumulative probability that a randomly generated watermark key (image and label) yields a watermark accuracy exceeding a specified threshold on an unmarked model.

Under the independence assumption, we generate 100 random watermarking keys and evaluate their label-matching accuracy on a set of 30 unmarked models.
The distribution of matching counts is approximated using a cumulative normal distribution, and the decision threshold is selected to correspond to a $p$-value of 0.05.
However, this technique provides only a rough estimate due to two limitations: 
each model has a unique classification probability distribution, even with the same architecture, owing to variations in training data and hyper-parameters; and (ii) the distribution of matching probabilities for random watermark keys is not explicitly modeled.
Section~\ref{sec:pvalues} has demonstrated the limitations in calculating the cumulative probability functions for small $p$-values.

\begin{table}[t]
    \centering
    \renewcommand{\arraystretch}{1.15}
    \resizebox{\linewidth}{!}{
    \begin{tabular}{c|p{0.62in}<{\centering}|p{0.5in}<{\centering}|p{0.5in}<{\centering}|p{0.5in}<{\centering}|p{0.5in}<{\centering}}
    \toprule
    \multirow{2}{*}{\bf Attack Types} & \multicolumn{5}{c}{\bf Robust (-) or Vulnerable (V) for CIFAR-10 / CIFAR-100} \\
    \cmidrule{2-6}
    {} & {\bf \TN{}} & {\bf Adi} & {\bf Content} & {\bf Noise} & {\bf Unrelated} \\
    \midrule
    {WM Ambiguity Attack} & {-/-} & {V/V} & {V/V} & {V/V} & {V/V} \\
    {Adaptive Denoising} & {-/-} & {-/-} & {-/-} & {-/-} & {-/-} \\
    {JPEG Compression} & {-/-} & {-/-} & {-/-} & {-/-} & {-/-} \\
    {Input Quantization} & {-/-} & {-/-} & {-/-} & {-/-} & {-/-} \\
    {Input Smoothing} & {-/-} & {-/-} & {-/-} & {-/-} & {-/-} \\
    {Adversarial Training} & {-/-} & {-/-} & {-/-} & {-/-} & {-/-} \\
    {Fine-Tuning (RTAL)} & {-/-} & {-/-} & {-/-} & {-/-} & {-/-} \\
    {Fine-Tuning (RTLL)} & {-/-} & {-/-} & {-/-} & {-/-} & {-/-} \\
    {Fine-Tuning (FTAL)} & {-/-} & {-/-} & {V/V} & {-/-} & {V/V} \\
    {Fine-Tuning (FTLL)} & {-/-} & {-/-} & {-/-} & {-/-} & {-/-} \\
    {Weight Quantization} & {-/-} & {-/-} & {-/-} & {-/-} & {-/-} \\
    {Weight Pruning} & {-/-} & {-/-} & {-/-} & {-/-} & {-/-} \\
    {Regularization} & {-/-} & {V/-} & {V/-} & {-/-} & {V/-} \\
    {Retraining} & {-/-} & {-/-} & {V/V} & {V/-} & {V/V} \\
    {Transfer Learning} & {V/V} & {V/V} & {V/V} & {V/V} & {V/V} \\
    \multirow{2}{*}{\shortstack{Cross-Architecture\\Retraining}} & \multirow{2}{*}{-/-} & \multirow{2}{*}{-/-} & \multirow{2}{*}{V/-} & \multirow{2}{*}{-/-} & \multirow{2}{*}{V/-} \\
    {} & {} & {} & {} & {} & {} \\
    {Adversarial Training} & {-/-} & {-/-} & {-/-} & {-/-} & {-/-} \\
    \bottomrule
    \end{tabular}
    }
    \caption{Robustness of different watermarking schemes against 17 attack types (threshold probability $p=0.01$).}
    \label{tab:robust}
    \vspace{-0.3in}
\end{table}

After evaluating each watermarking scheme against 17 distinct attack types, we present their robustness in Table~\ref{tab:robust}. 
The results show that \TN{} is resistant to the watermark ambiguity attack, whereas all four baseline schemes are vulnerable.
Also, against the remaining 16 attacks, \TN{}, \emph{Adi}, and \emph{Noise}-based scheme exhibit relatively higher robustness.
However, compared to \TN{}, \emph{Adi} is vulnerable to regularization attacks, and \emph{Noise}-based scheme is vulnerable to retraining attacks. 
With the exception of transfer learning where the $p$-value is 0.012/0.035 on CIFAR-10/CIFAR-100, \TN{} consistently achieves $p$-values between $6 \times 10^{-3}$ and $1 \times 10^{-8}$ across all other watermark removal attacks.
Thus, \TN{} is the most robust scheme that is able to resist multiple watermark removal attacks.

\vspace{-0.03in}
\section{Discussion}

\subsection{{Usability}}

The \TN{} scheme can be adopted by model owners (whether commercial vendors or individual developers) to protect their intellectual property. 
Based on dynamic watermarking, \TN{} does not interfere with the model’s primary functionality, since the triggers are out-of-distribution inputs that resemble random noise and do not affect regular inference.
We tried various hash functions (e.g., MD5, SHA1, SHA128) and observed the selection does not affect the final results. 
Block cipher in counter mode could be an alternative; however, hashing is faster and the hard-bound is not an issue.
Moreover, \TN{} introduces negligible training overhead since the number of triggers is not comparable to the size of original training dataset.
Besides, cryptographic chains can be extended to watermark RNNs and LLMs, but in different formats.

\TN{} provides a higher security guarantee due to its higher marginal utility. 
Specifically, for the ResNet-18 model trained on CIFAR-10, 14 matches out of 100 triggers are sufficient to support a successful ownership claim. 
For models trained on CIFAR-100, even fewer matches are required, as the threshold is dependent on the output dimension. 
A larger output space corresponds to a lower random guessing probability, thereby allowing a more relaxed matching requirement.
Besides, the watermarks can be embedded through fine-tuning, using initial weights and learning rates that differ from those used in ``training-from-scratch".

In practical watermark verification, it is not necessary to disclose the entire key chain. 
For example, for a model trained on CIFAR-10, ownership can be verified by presenting only the first 20 trigger inputs ($B_1$ to $B_{20}$). 
If at least 14 of 20 triggers match, the ownership claim is considered valid.
The remaining triggers can be reserved for further verification rounds. 
Also, due to the one-way property, any unused triggers in key chain remain secure and undisclosed.

\vspace{-0.03in}
\subsection{Scalability}
To extend \TN{} to larger and more complex datasets (e.g., ImageNet), several adjustments are required to the watermark configuration.
Although the increased class number reduces the likelihood of a single successful guess, the probability of accidentally achieving the minimum match threshold may not decrease proportionally and can even increase.
Thus, the chain length $L$ should be increased accordingly, but kept sufficiently short to preserve the nature of out-of-distribution.
In addition, the digital signature should be encoded in a higher-base numeral system to match the class number.
Due to the higher dimension of ImageNet data, hash-like triggers are more likely to be memorized, as they reside in sparser regions of the data manifold, far from natural image distributions.

Hyperparameters should be selected carefully to balance watermark security, robustness, and efficiency. We recommend setting $L \geq \sqrt{C}$, where $C$ is the output class number, to ensure sufficient watermark entropy. The $p$-value, which determines the Hamming distance threshold, should range between $10^{-2}$ and $10^{-6}$, depending on the desired security level. While the specific hash function has limited impact on cryptographic strength, we recommend strong cryptographic hashes (e.g., SHA-256) for high-security applications.

\vspace{-0.03in}
\subsection{Limitations and Future Work}

The main contribution of ChainMarks is to defeat watermark ambiguity attacks, which are emerging threats against all existing DNN model watermarking methods. 
ChainMarks cannot effectively defeat removal attacks via transfer learning and knowledge distillation. 
In fact, none of the existing watermarking techniques is robust against these two methods.
Besides, \TN{} focuses on watermarking the classification models with the inputs of images. However, our idea of watermarking with a key chain can be extended to other input formats or other modeling tasks.
For text data, it is feasible to convert a hash value into a word (i.e., word ID)~\cite{li2023plmmark} or a pseudorandom string~\cite{wang2024dye4ai}.
For graph-based input, we can transform a binary hash value into an adjacency matrix to generate graph-structured data.
We leave this topic to future work.

\vspace{-0.03in}
\section{Related Work}

\subsection{Backdoor Poisoning Attacks}
Backdoor poisoning attack is a special case of targeted poisoning attacks that maintain overall performance and induce misbehaviors in triggers~\cite{severi2021explanation}. 
Data manipulation is the main technique for backdoor poisoning attacks~\cite{tian2022comprehensive}. Adversaries introduce either visible~\cite{gu2017badnets, li2021backdoor} or invisible~\cite{li2021invisible, pan2022hidden, liu2022loneneuron, quiring2020backdooring} patterns into poisoning samples.
Also, triggers for poisoned samples can be generated by optimization to achieve better performance~\cite{liu2018trojaning, zhao2020clean, li2020invisible}. Semantic backdoor attacks leverage the semantic part of the samples as trigger patterns, so it is unnecessary to modify the input at inference time~\cite{bagdasaryan2021blind, bagdasaryan2020backdoor}. 
Similarly, a hidden backdoor can be activated by combining certain objects in images~\cite{lin2020composite}. It is possible to conceal triggers using image scaling attacks~\cite{xiao2019seeing}.
However, almost all backdoor attacks are sample-agnostic and therefore can be defeated by trigger-synthesis-based defenses~\cite{wang2019neural} and saliency-based defenses~\cite{chou2020sentinet}. Hence, sample-specific backdoor attacks are proposed to contain different trigger patterns for different poisoned samples~\cite{nguyen2020input, li2021invisible}. Backdoor attacks can also be launched in the physical world using a pair of glasses~\cite{chen2017targeted} or a post-it note~\cite{gu2017badnets}. In addition, backdoor attacks can be applied in different fields, e.g., computer vision~\cite{jia2022badencoder, yao2019latent, noppel2022disguising}, natural language processing~\cite{chen2021badnl}, speech recognition~\cite{zhai2021backdoor}, software code~\cite{yang2022jigsaw}, and graph learning~\cite{xi2021graph,zhang2021backdoor}.

\vspace{-0.03in}
\subsection{DNN Watermarking Schemes}

\noindent{\bf White-box/Black-box/Box-free Watermarking.}
Based on the information accessible during watermark verification, watermarking schemes can be classified as white-box, black-box, and box-free~\cite{li2022backdoor}.
A white-box watermarking scheme grants users or adversaries access to the internal information of DNNs, e.g., model structures, model weights, and hyperparameters~\cite{kuribayashi2021white, lv2023robustness, uchida2017embedding}.
However, due to the strong assumption, white-box watermarking has a larger capacity but limited applicability~\cite{chen2019blackmarks}.
Black-box watermarking schemes only allow users to access the final outputs of DNN models by feeding a set of inputs~\cite{ adi2018turning, le2020adversarial}, allowing IP protection for Machine Learning as a Service (MLaaS)~\cite{kapusta2021protocol}.
Box-free watermarking is similar to the black-box one; however, this mechanism is only applied to DNNs with high-dimensional outputs, i.e., image processing models, since watermarks can be embedded into the outputs for any inputs by a high output entropy~\cite{zhang2020model, wu2020watermarking}.

\noindent{\bf Static vs. Dynamic Watermarking.}
Watermarking schemes can be classified as static or dynamic based on watermarking methods~\cite{barni2021dnn}.
Static methods embed watermarks in the static DNN parameters that are not changed during the operation.
For example, watermarks can be embedded as the probability distribution of weights~\cite{chen2019deepmarks} or model weights~\cite{uchida2017embedding, tartaglione2021delving}.
However, most static methods imply white-box watermarking, since the model parameters need to be accessible during verification.
Dynamic methods associate watermarks with the network behaviors in correspondence to some specific inputs~\cite{adi2018turning, rouhani2019deepsigns, szyller2021dawn, le2020adversarial, zhang2018protecting, guo2018watermarking, zhang2020model}.
A set of secret inputs/patterns with target labels (i.e., triggers) are carefully crafted as the watermarks to be ingrained into the DNN in the training process along with the original data set. 
In watermark verification, model behaviors will be tested to verify the presence of the watermark. 
Dynamic watermarking usually generates trigger input by leveraging DNN backdoor poisoning attacks~\cite{li2023black}.
However, dynamic watermarking does not imply black-box watermarking, as it can also be used as white-box watermarking.
For example, Rouhani et al. use activation maps to embed and verify watermarks~\cite{rouhani2019deepsigns}.

\noindent{\bf Zero-bit vs. Multi-bit Watermarking.}
Based on the type of watermark contents disclosed in the verification, watermarking schemes can be classified as zero-bit and multi-bit~\cite{li2023universal}.
In zero-bit watermarking, only the watermark presence is detected~\cite{adi2018turning, le2020adversarial, zhang2018protecting}; while in multi-bit watermarking, both the watermark and its presence should be present in the verification process~\cite{chen2019blackmarks}.
A multi-bit watermarking scheme can be converted to a zero-bit one.

\vspace{-0.03in}
\subsection{Attacks/Defenses on DNN Watermarks}

DNN watermark attacks include model modification attacks, evasion attacks, and active attacks~\cite{xue2021dnn}.

\noindent {\bf Model Modification Attacks.}
Model weights are often modified by the pirate. Model modification attacks include model fine-tuning~\cite{uchida2017embedding, chen2019leveraging, chen2021refit}, model pruning or parameter pruning~\cite{rouhani2019deepsigns}, model weight compression~\cite{uchida2017embedding}, and model retraining~\cite{cao2021ipguard, namba2019robust}.

\noindent {\bf Evasion Attacks.}
Evasion attacks are more complicated.
Shafieinejad et al. investigate the removal of backdoor-based watermarks with white-box, black-box, and inference attacks~\cite{shafieinejad1906robustness}.
Also, DNN laundering is shown to reset backdoor watermarks~\cite{aiken2021neural}.
Attackers can manipulate a model to remove the owner's signature if watermark presence is known in advance~\cite{guo2018watermarking}.
Reverse engineering can be used if the original training dataset is obtained~\cite{fan2019rethinking}.
Liu et al. propose a data augmentation scheme to mimic the backdoor trigger behaviors~\cite{liu2021removing}.
Gong et al. dynamically adjust the learning rate to purify backdoors~\cite{gong2023redeem}.
Attackers can leverage resource-efficient attacks~\cite{ilyas2018black, tu2019autozoom, suya2020hybrid} to find black-box adversarial examples.
By combining hybrid attacks with seed prioritization, adversarial examples can be obtained using only a few queries.
The two main optimization techniques used in attacks are AutoZOOM~\cite{tu2019autozoom} and NES~\cite{ilyas2018black}.

\noindent {\bf Active Attacks.}
Ambiguity attack tends to forge an additional watermark on the DNN model to doubt the ownership verification~\cite{fan2019rethinking, chen2023effective}.
Also, several methods are proposed to detect watermarks for further attacks~\cite{xu2017feature, ma2022beatrix, cheng2023beagle}.
Attackers can overwrite the watermarks if they know the watermarking method~\cite{chen2019deepmarks, chen2019blackmarks, rouhani2019deepsigns}.
In addition, it is possible to prevent the copyright owner from verifying the ownership by using a watermark collusion attack~\cite{chen2019deepmarks}. Similarly, attackers can also detect and modify the watermark query to prevent watermark verification~\cite{namba2019robust}.

\noindent {\bf Countermeasures.}
Entangled watermarks increase the similarity between watermarks and task features, improving resistance to trigger detection~\cite{jia2021entangled}; however, \TN{} prevents attackers from adding new ambiguous watermarks.
DynaMarks defeats model extraction attacks by dynamically changing the responses of the model's prediction API during the inference phase~\cite{chakraborty2022dynamarks}.
Bansal et al. propose randomized smoothing to improve the difficulty of watermark removal attacks~\cite{bansal2022certified}.
To defeat watermark ambiguity attack, Fan et al. propose a passport layer so that model performance deteriorates due to forged signatures~\cite{fan2019rethinking}.
\vspace{-0.05in}
\section{Conclusion}

We propose a new DNN watermarking scheme, \TN{}, which is resistant to watermark ambiguity attacks by introducing cryptographic constraints among watermark triggers and target labels, along with the model owner’s digital signature. 
Experiments show that \TN{} exhibits higher or comparable levels of resistance compared to other watermark schemes against various watermark attacks, including input processing, model modification, and model extraction attacks.
To determine watermark decision threshold, the proposed two-phase Monte Carlo method shows its accuracy and applicability across a range of watermarked DNN models.
The marginal utility of \TN{} is higher than that of the other schemes, providing a higher probability guarantee of the watermark presence in the DNN models with the same level of watermark accuracy.

\vspace{-0.05in}
\begin{acks}
This work is partially supported by Office of Naval Research grant N00014-23-1-2122 and Commonwealth Cyber Initiative award N-3Q25-002.
\end{acks}
\bibliographystyle{ACM-Reference-Format}
\bibliography{reference}
\end{document}